\documentclass[12 pt]{article}
\usepackage[utf8]{inputenc}
\usepackage{fullpage}
\usepackage{afterpage}
\usepackage{graphicx}
\usepackage[round]{natbib}
\usepackage{setspace}
\usepackage{amsmath, amsthm, amssymb}
\usepackage{latexsym}
\usepackage{sectsty}
\usepackage{color}
\usepackage[usenames,dvipsnames]{xcolor}
\usepackage{indentfirst}
\sectionfont{\large}
\subsectionfont{\normalsize}
\usepackage{url}

\newtheorem{theorem}{Theorem}
\newtheorem{lemma}{Lemma}
\newtheorem{proposition}[theorem]{Proposition}

\renewenvironment{proof}[1][Proof]{\begin{trivlist}
\item[\hskip \labelsep {\bfseries #1}]}{\end{trivlist}}

\newcommand{\R}{\mathbb{R}}
\newcommand{\E}{\mathbb{E}}
\newcommand{\ol}{\overline}

\renewcommand{\qed}{\nobreak \ifvmode \relax \else
      \ifdim\lastskip<1.5em \hskip-\lastskip
      \hskip1.5em plus0em minus0.5em \fi \nobreak
      \vrule height0.75em width0.5em depth0.25em\fi}

\title{
Elections, Protest, and Alternation of Power\footnote{
Previous versions of this paper were presented at the 2nd Annual Meeting of the European Political Science Association and the Elections and Political Order conference at Emory University. Many thanks to participants in these seminars, as well as Jos\'{e} Fern\'{a}ndez Albertos, Scott Gelbach, Alastair Smith, and Milan Svolik for comments and suggestions.}
}
\author{Andrew T. Little\footnote{Department of Politics, New York University, andrew.little@nyu.edu.}
\and
Joshua A. Tucker\footnote{Department of Politics, New York University, joshua.tucker@nyu.edu.}
\and
Tom LaGatta\footnote{Courant Institute of Mathematical Sciences, New York University, lagatta@cims.nyu.edu}
}

\date{January 2013}

\begin{document}
\begin{titlepage}

\maketitle

\thispagestyle{empty}
\begin{abstract}
	Despite many examples to the contrary, most models of elections assume that rules determining the winner will be followed. We present a model where elections are solely a public signal of the incumbent popularity, and citizens can protests against leaders that do not step down from power. In this minimal setup, rule-based alternation of power as well as ``semi-democratic'' alternation of power independent of electoral rules can both arise in equilibrium. Compliance with electoral rules requires there to be multiple equilibria in the protest game, where the electoral rule serves as a focal point spurring protest against losers that do not step down voluntarily. Such multiplicity is possible when elections are informative and citizens not too polarized. Extensions to the model are consistent with the facts that protests often center around accusations of electoral fraud and that in the democratic case turnover is peaceful while semi-democratic turnover often requires citizens to actually take to the streets.
\end{abstract}
\end{titlepage}

\doublespacing
\indent

\section{Introduction}

Why do incumbent politicians ever cede office voluntarily?  After all, most models of voting begin with the assumption of that politicians are office seeking, so why do we then assume that these same actors will simply give up power because of an election result?  Of course, we know incumbents often do give up power after elections, as George H.W. Bush, Nicholas Sarkozy, Slobodan Milosevic and Eduard Shevardnaze can readily attest.  However, most models tend to separate the mechanism by which the former two relinquished power (voluntarily after losing an election) from the latter two (at the behest of protesters in the street following fraudulent elections that they or their party ostensibly won).
% \footnote{Other examples of recent fraudulent elections that the incumbent officially won but still resulted in turnover in office include Peru (2000), Ukraine (2005), Kyrgyzstan (2005), and the Ivory Coast (2010).  Recent fraudulent elections that have not (yet) led to the downfall of the ruling regime but did result in substantial protests include Iran (2009), Zimbabwe (2010) and Russia (2011).}
Highlighting the importance of this distinction, when Shevardnaze's successor Mikheil Saakashvili admitted defeat for his party following Georgia's 2012 parliamentary election, it was hailed as the first democratic transfer of power in the Caucasus. 

In this paper, we present a single model that encompasses both of these types of incumbent turnover. To do so, we treat elections not as a binding contract, but solely as a public signal of incumbent popularity.
% \footnote{Although the language of contracts is not often employed in discussing election, most formal models of elections do in fact treat elections as a binding contact by which incumbents agree \emph{a priori} to relinquish power should they fail to receive a certain number of votes; why incumbents do in fact go on to relinquish power in these cases is rarely discussed or modeled.}
The signal -- which contains noise from multiple sources and can include manipulation via electoral fraud -- is observed by citizens and the incumbent, who then decides whether to stay in office or relinquish power voluntarily. If not, citizens then have the opportunity to take to the streets to try to force the incumbent to step down. 
% By generating public information, the election result influences individual citizens' beliefs about whether \emph{other} citizens dislike the regime enough to protest, and the incumbent's belief about the value of attempting to stay in office. 

Our main result is that -- despite the minimal role assigned to the election -- both rule based (or democratic) alternation in power and non-rule-based (or semi-democratic) alternation of power can all arise in equilibrium \emph{from the same model}. By rule-based alternations in power, we are referring to cases in which the incumbent stays in power only if a formal electoral threshold is reached.  However, this is not because the incumbent inherently respects the rule, but because such behavior \emph{emerges endogenously in equilibrium}.  This is possible when citizens can choose between multiple equilibria for some election results, one with a ``high-protest" strategy and another with a ``low-protest" strategy.  While playing the high protest strategy and inducing the incumbent to step down if and only if the election result is above the legal threshold is not the only equilibrium that could be chosen, it is a convenient focal point for selecting among multiple equilibria. Moreover, if a codified legal threshold (e.g., $50\%$ of the vote) acts as a focal point determining when citizens play the high protest strategy, then the elections can appear \emph{democratic}: a candidate who ``loses" the election (i.e., receives a share of the vote below $50\%$) will cede power without protest.\footnote{Of course, compliance with electoral rules is not a sufficient condition for democracy, but it is certainly necessary and arguably the most important requirement beyond holding elections in the first place.}  A key finding from our model is that these rule-based alternations in power tends to be possible when elections are very informative and citizens not too polarized.

When these conditions do not hold, there is a unique equilibrium in the protest stage, which precludes the focal point effect required for rule-based alternation. Still, the incumbent may step down in this case if the anticipated level of protest is sufficiently high, and this will tend to happen when a lower than expected election result generates public information that citizens are unhappy with the regime. That is, the incumbent may step down if the election result is below a critical threshold, but in a ``non-rule based" manner as this critical threshold need not correspond to a legally defined electoral rule.  Intuitively, this can describe a situation in which an election is held, an incumbent announces that she has indeed cleared the legal threshold for re-election, protest ensues, and the incumbent ends up stepping down despite a ``winning" election result.  We can think of this as a form of turnover in ``semi-democratic" systems: elections are held, but turnover is neither ensured nor prohibited based on the announced result of the election. Instead the key determinant of turnover is whether citizens take to the streets following an election or whether, anticipating that citizens will take to the streets (or continue protesting), the incumbent steps down.

In addition, we present two extensions of the model which generate results consistent with important features of post-election protest. First,  we incorporate uncertainty about how much fraud was committed, but give citizens a public signal of the level of fraud, which could correspond to media coverage or reports from international monitoring groups. Signals indicating high levels of fraud make citizens more apt to take to the streets, consistent with observed behavior.  Second, we allow the incumbent to step down either before or after the protest. In the democratic equilibrium, and losers are sure to face high protest if the do not yield and hence step down right away. However, in the equilibrium with semi-democratic turnover, the incumbent tends to be more uncertain about how much protest there will be and will often ``wait things out'', only stepping down if the protest is in fact large. This distinction is consistent with the fact that democratic turnover tends to be peaceful while semi-democratic turnover often requires citizens to actually take to the streets. 
 
The rest of the paper is organized as follows. Section \ref{sec.past} briefly reviews related work. Section \ref{sec.model} provides a basic overview of the model. Section \ref{sec.unique} analyzes the unique equilibrium case, and Section \ref{sec.multiple} the multiple equilibrium case. Sections \ref{sec.monitor} and \ref{sec.beforeafter} present the extensions incorporating uncertainty about the level of fraud and the option to step down before or after protest,  and section \ref{sec.conclusion} concludes.

\section{Past Work}
\label{sec.past}

The possibility or realization of leaders relinquishing power following an election plays a central role in prominent theoretical and operational definitions of democracy \citep{P1991, PACL2000, CGV2010}. However, most formal theories of democratization do not emphasize alternation of power, but decisions such as expanding the franchise \citep{AR2000, BDMS2009} or the decision to hold elections \citep{C2009, ielec}. Models that do consider the decision to step down from office implicitly rely on equilibrium selection for the enforcement of electoral rules \citep{P2005, F2011}, whereas here we focus on the existence of multiple equilibria explicitly. 
% This allows for a closer connection between modeling democratic and nondemocratic elections, which have been generally treated as separate research questions \citep{GLO2009}.

Our model builds on recent work that analyzes various aspects of elections by treating them as public signals \citep{LV2006, C2009, ES2011, ielec, ggelec}.\footnote{We do not explicitly model why incumbent regimes invite monitoring \citep{K2008, H2012, fmne}, choose a particular level of fraud \citep{GP2009, S2011, fmne}, or hold elections or set up legislatures in the first place \citep{G2006, Gan2008, ielec, GK2011}, as these questions have been addressed elsewhere. Further, fraudulent elections with monitoring are so pervasive in contemporary countries that we are not losing many examples by considering a model where elections occur that simultaneously feature both fraud and some sort of election monitoring \citep{MK2009, GLO2009, S2011, H2012}.} Several recent models examine post-election protest in an informational framework \citep{Kuhn2012, SC2012, R2012}, but treat the opposition and/or citizenry as a unitary actor, abstracting away from the coordination problem central to our argument.

\citet{ES2011} and \citet{ggelec} begin with the same premise as we do here -- where an election result is a public signal in an incomplete information coordination game -- but only focus on the unique equilibrium case and do not allow the incumbent to step down to avoid protest, and hence there is little overlap with our main results. Models with incomplete information and coordination -- generally called ``global games'' -- are becoming common in the study of political phenomena such as riots \citep{A2000} and revolutions \citep{BDM2010a, SB2011}.\footnote{There is no commonly accepted technical definition for what constitutes a global game. To point to a common reference, the payoff structure of our model meets the listed assumptions in section 2.2.1 of \citet{MS2003}, but given the prior and election result citizens do not have a laplacian belief about the incumbent popularity before receiving their private signal. In fact, this public information is precisely the reason why we do not always obtain uniqueness unlike the main result in that section.} 

Modeling the post-election strategic interaction with a global game is particularly valuable in our setting because such games may or may not have multiple equilibria, and it is generally straightforward to derive conditions under which there are multiple equilibria. In particular, we argue that coordination and multiple equilibria are central to understanding why electoral rules can be enforced by equilibrium behavior, a point made by various authors about the rule of law in general \citep{C1995, W1997, MMP2001, H2003, M2008}.\footnote{\citet{ielec} argues for the importance of multiple equilibria in understanding consolidated democracy, but takes the multiplicity as a given and treats noncompetitive elections and elections resembling consolidated democracy in different models.}  Thus the model here is the first along these lines to consider a game that sometimes has a unique equilibrium, with a focus on when multiple equilibria are present and hence when particular rules (in our context, electoral rules) are enforceable. 

%Morris and Shin (2003) describe the global games approach as releasing analysts from a ``straightjacket'' of multiple equilibria. However, in the context of elections, it is the unique equilibrium case that creates a straightjacket precluding the construction of equilibria resembling consolidated democracy.

The relationship between elections, electoral fraud, and protest has been the focus of a good deal of scholarly interest following the run of ``colored revolutions" in post-communist countries over the last decade \citep{BW2011}.  In particular, \citet{T2007} illustrates how electoral fraud can be a potent focal point for solving collective action problems among citizens living under abusive regimes, thus helping to justify our focus in this model on \emph {post-election} protests.\footnote{Meirowitz and Tucker (2013) add a dynamic element to this approach, but they model the protester as a single agent and not as a collective of agents as we do here.} Moreover, in comparative assessments of the factors that lead to successful colored revolutions, multiple authors highlight the importance of election results as either galvanizing or deflating an opposition movement \citep{M2005, BW2011}.  Even more recently, the reaction of the Russian opposition to Putin's unexpectedly strong election results in 2012 (as opposed to weaker results in the 2011 parliamentary elections from the ruling United Russia party) has been noted by academic bloggers as having a potential deflationary effect on the nascent Russian opposition movement.\footnote{http://themonkeycage.org/blog/2012/03/05/russia-2012-presidential-election-post-election-report/}

\section{The Model}
\label{sec.model}

The actors in the model are an incumbent denoted $I$, and $N$ citizens, indexed by $j$.\footnote{We could also think of this group of citizens as a collection of elites who have the power to oust the regime or bring other citizens (followers) out into the street. For the sake of clarity, however, we refer to these people simply as citizens.} Notation referring to the incumbent will often have a superscript $I$, and notation for the citizens will generally have a superscript $C$. We analyze the model for a finite but arbitrarily large number of citizens, i.e., as $N \to \infty$. Practically speaking, even large post-election protests only include a small fraction of the total population, so the citizens modeled here are better conceptualized as those that could plausibly protest.\footnote{Consider for example the recent Russian post-election protests of 2011-12.  Even if we accept the high end of estimates for the number of Russians who participated in these protests at between 250,000-300,000 people, that is still less than 1\% of the population of a country with over 140 million people.}

The sequence of moves is:
\begin{center}
	\includegraphics{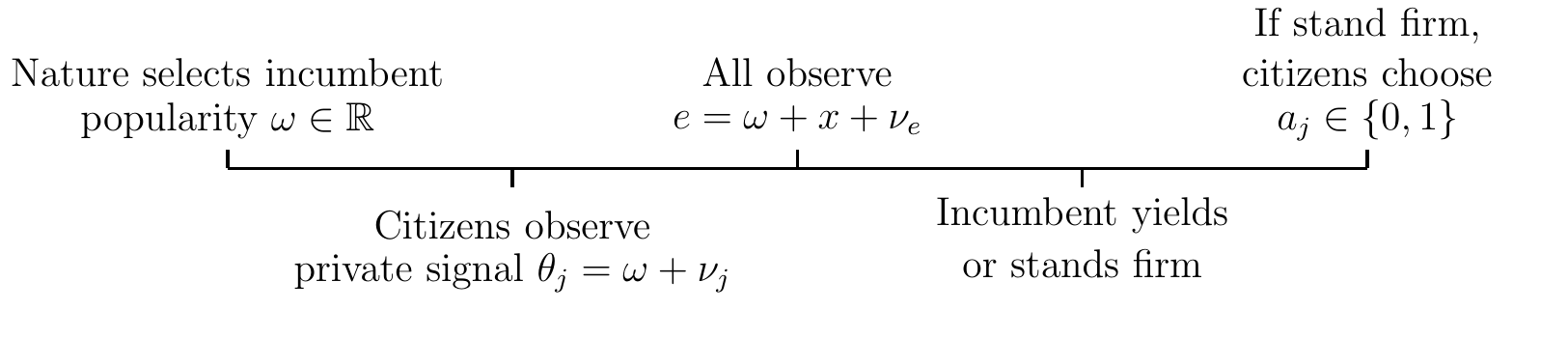}
\end{center}

We break down a citizen's support for (or approval of) a current leader/government into two factors: the average assessment of the ``performance" of the leader -- analogous to an approval rating in a public opinion poll -- and an idiosyncratic component capturing whether or not that particular individual likes the leader more or less than average.  Thus we model each citizen as having a personal level of \emph{regime sentiment} $\theta_j$, which is written as the average incumbent \emph{popularity}, denoted $\omega$, plus an idiosyncratic term that varies by citizen ($\nu_j$). As formalized below, citizens with negative or anti-regime sentiment ($\theta_j < 0$) will generally want to protest to force the leader to step down if necessary, while those with positive or pro-regime sentiment ($\theta_j > 0$) will not want to protest against the regime. As citizens only directly observe their own regime sentiment, we will often refer to $\theta_j$ as citizen $j$'s \emph{private signal}. At the beginning of the game, all actors share a common prior belief about the average popularity (specified below).

The election result ($e$) is simply a public signal of the average popularity ($\omega$). The result is contaminated by fraud (should it occur), denoted $x$, and random noise $\nu_e$. The random component accounts for any factors that might affect the election result independent of incumbent popularity and electoral fraud, e.g., uncertainty over how closely those turning out to vote resemble the population at large (or the potential protesters).  For now, the amount of fraud $x$ is fixed and known; an extension will examine uncertainty about the amount of fraud. Assume the final result takes an additive form: $e = \omega + x + \nu_e$.\footnote{An unsatisfying aspect of this specification is that the election result can be any real number, as opposed to representing something more concrete like the incumbent vote share on $[0,1]$. However, all of the analysis here would apply (with clunkier notation) should the citizens instead observe a vote share $s(e)$ where $s: \R \to [0,1]$ is a strictly increasing function.} 

After observing the election result, the incumbent then chooses whether to step down. The incumbent makes this decision regardless of whether or not she ``wins'' the election.\footnote{We refer to the incumbent with the pronoun ``she'' and citizens with ``he.''} In fact, a major feature of the model is that we do not assume any notion of winning into the payoffs or any other aspect of the model; the election is simply a signal of incumbent popularity, and a noisy one at that. If the incumbent does not step down, the citizens decide whether or not to protest. Let $a_j$ denote the decision to protest or not for citizen $j$, where $a_j = 1$ means protesting and $a_j = 0$ not protesting. Denote the proportion of protesting citizens with $\rho \equiv \sum_{j=1}^N a_j/N$. 

%As elaborated below, a threshold determining whether the incumbent steps down in a manner resembling an electoral rule determining who wins may arise from equilibrium behavior.

The incumbent payoff is:
\[
u^I(\rho) = \begin{cases}
	y^I & \text{if stepping down}\\
	1 - \rho & \text{otherwise}
\end{cases}
\]

That is, $y^I$ is the incumbent's payoff for yielding; if she does not yield, she instead gets a payoff that is linearly decreasing in the size of protests.\footnote{Here and throughout the paper we use ``linear'' when ``affine'' would be more precise. That is, for the incumbent utility as a function of $\rho$ to be linear in the strict sense it would need to be the case that $u^I(0)=0$, which is not true.} Assume $0  < y^I < 1$, so  there are some levels of protest where the incumbent prefers to step down and some where she prefers to stand firm.

%% A new paragraph here. Could maybe use some Joshy touch. DONE

We see citizens as having the following traits. First, they want a high level of protest against unpopular leaders (i.e., they want to see these leaders out of office). Second, they are more likely to join these protests the more people they believe will take part, both because of the information conveyed by larger numbers of protesters (i.e., that the incumbent is less popular) and because the cost to joining a larger protest is lower. Third, citizens also enjoy protesting (i.e., get an ``expressive'' payoff) against leaders that they personally dislike. Finally, we assume that citizens are uncertain about how much other citizens dislike the regime and hence their propensity to protest.  It is against this backdrop of uncertainty that the election result provides  information about the beliefs of others by providing a signal of incumbent popularity. Formally, then, the citizen payoffs are:\footnote{Some important aspects of protest that we abstract from include the role of opposition elites' strategies \citep{BW2011}, how events other than elections can signal incumbent popularity \citep{T2007}  or the possibility of learning across successive rounds of protest in a single country \citep{MT2013} or cross nationally \citep{BW2011, A1999}.}
\begin{align*}
	u^C (a_j; \theta_j, \rho, \omega) = \begin{cases}
		y^C & \text{incumbent steps down} \\
		v(\omega)(1 - \rho)   & \text{incumbent stands firm and }a_j = 0\\
		v(\omega)(1 - \rho) - c - k_1 (1 - \rho) - b \theta_j& \text{incumbent stands firm and } a_j = 1
	\end{cases}
\end{align*}

If the incumbent steps down the game ends, and all citizens get a ``yielding'' payoff $y^C$. If the incumbent does not step down, citizens get a partial payoff of $v(\omega) (1 - \rho)$ regardless of the action taken, where $v(\omega)$ is increasing in $\omega$ and can be positive or negative.
% \footnote{That is, there is some value $\omega_0$ for which $v(\omega_0) = 0$.}
This captures the assumption that all citizens want there to be high level of protest when the incumbent popularity is low ($v(\omega) < 0$) and want there to be a low level of protest when the incumbent popularity is high ($v(\omega) > 0$).  This does not affect the equilibrium strategies when $N \to \infty$, as the individual citizen's effect on $\rho$ becomes negligible. However, this term will play an important role in the analysis of why a majoritarian electoral rule -- when possible -- may be a particularly good equilibrium for the citizens.

The $c > 0$ term is a fixed cost of protest independent of how many citizens join. That is, $c$ is intended to reflect costs such as lost wages or the discomfort of standing in the cold. The $k_1 (1 - \rho)$ term is a variable cost term, which captures the idea that it is safer to join a larger protest.  For example, Bashar Assad's regime in Syria reportedly ordered snipers to shoot 10 participants per protest in the recent uprising; with this number fixed increasing the denominator made coming out to the streets less dangerous.\footnote{See http://www.juancole.com/2012/08/the-age-of-mass-killing-comes-to-syria-france-pushes-govt-in-exile.html}

% Andrew: I always found the $-b \theta_j$ discussion confusing, so I've rewritten this to introduce $b$ first 

Finally, citizens get an ``expressive'' benefit of protest equal to their regime sentiment multiplied by a negative constant $-b < 0$. When $\theta_j$ is negative -- i.e., citizen $j$ dislikes the regime -- this payoff is positive and hence all else equal citizen $j$ wants to protest, when $\theta_j$ is positive he prefers to not protest.
% \footnote{It may be more natural to think of the $b$ parameter as varying by citizens: consider perhaps an unemployed twenty-one year old who sees protest as a break from a monotonous life and an employed 50 year old who sees the protest as a necessary act but which will cost her time away from work and her family.  Both may dislike the regime equally, but the former gets more utility out of the protest. Still, our formulation allows for citizens to be more extreme than others based on their private signal, and we do not see how a more complicated version would change the results or lead to new insight.}
 As high values of $b$ have a multiplicative effect on underlying differences in how much citizens like or dislike the incumbent -- or, loosely speaking, how much they dislike the incumbent relative to opposition parties -- a useful way to think about this parameter is that it captures how \emph{polarized} society is.
% \footnote{Another way to interpret this is in terms of how motivated the group of potential protesters are to actually participate in politics: are they reluctant protesters or are they particularly eager?}

We include this somewhat nonstandard expressive payoff for a combination of technical and substantive reasons. Technically, the possibility of a unique equilibrium relies on ``two-sided limit dominance'': i.e., citizens who sufficiently dislike the regime always protest and those who sufficiently like the regime never protest. A more common way to attain this property is for citizens receiving extreme signals to be certain that the protest will ``succeed'' by reaching some critical mass even if none of the citizens modeled themselves join.\footnote{\citet{BDM2011} uses a different payoff structure where uncertainty about payoffs -- as used here -- only results in one-sided limit dominance. This is not because the uncertainty is about payoffs \emph{per se}, but because no signal gives payoffs that give a dominant strategy to participate (protest) as it can here.}   However, in the context of protesting after an election, we find it less intuitively plausible that citizens with extreme anti-regime beliefs protest because they think the regime is inevitably going to fall rather than because they are simply extremely unhappy with the regime and willing to take a costly action in pursuit of bringing down the regime.

%New paragraphs to address concerned raised by Scott, Jim, Milan, etc. It needs to be a bit technical as it is to placate the technical folks. Could maybe move it though?

%JT: I switched this around a little bit, moving the previous two paragraphs up to come immediately after introduction of expressive payoff - I think it flows better this way.  Then I moved the remaining "here's stuff we don't do" paragraph (which I also exanpded a bit) into a footnote which I  put right after first sentenced of experessive payoff paragraph, i.e., right after we introduce the final component of the final utility function. I thought it was kind of trivial to have in the text, but you can put it back here if you want.

We solve for Perfect Bayesian Equilibria with some additional common restrictions elaborated below. Relying on sequential rationality, we first we solve the ``protest stage'' that occurs if the incumbent does not step down, and then determine when the incumbent steps down.\footnote{As the incumbent has no private information, citizens do not make inferences about her popularity from her decision.}

\subsection*{The Protest Stage}

Formally, the citizen strategy is a mapping from the election result and her private signal (i.e., personal anti-regime sentiment) to the decision to protest or not. As is standard, we assume at the outset that citizens use a symmetric strategy of the natural form ``protest if and only if $\theta_j < \hat{\theta}(e)$''; i.e., if and only if their personal distaste for the regime is sufficiently strong.
% \footnote{Such equilibrium strategies arise naturally from the symmetries of our model; see \citet{MS2003} for an overview.} 
Crucially, this threshold of distaste may be affected by the election result.

When decided whether to take to the streets, each citizen is uncertain about $\rho$: as the other citizens' regime preferences ($\theta_j$'s) are private, he does not know how many others dislike the regime enough to protest. Given the linearity of citizen payoffs, the protest decision depends only on the \emph{expected} proportion of protesters; 
% (that is, the belief regarding how many \emph{other} people will protest)
we denote this as $\E[\rho|\cdot]$. As $N \to \infty$, this proportion is unaffected by individual decisions, i.e.,  $\lim_{N \to \infty} \E[\rho|a_j = 1, \cdot] = \lim_{N \to \infty} \E[\rho|a_j = 0, \cdot]$.\footnote{In other words, no citizens are ``pivotal'' in the outcome of the protest, see \citet{smith2011} for a related model with more emphasis on pivotality.} So as $N \to \infty$ citizen $j$ protests in equilibrium  if and only if:
\begin{align}
\notag	v(\omega)(1 - \E[\rho|\cdot]) &\leq v(\omega)(1 - \E[\rho|\cdot]) - c - k_1 (1 - \E[\rho|\cdot]) - b \theta_j\\
\label{eqmtheta}	\theta_j &\leq \frac {-k_1 - c + k_1 \E[\rho|\cdot]} {b}
\end{align}
Since the expected level of protest must be between 0 and 1, we can place bounds on the RHS of equation \ref{eqmtheta}. So, citizens that sufficiently dislike the regime ($\theta_j \leq \frac {-k_1 - c} {b}$) have a dominant strategy to protest and those that sufficiently like the regime ($\theta_j \geq \frac {-c} {b}$) have a dominant strategy to not protest. That is, some citizens who will protest and some who will not protest regardless of their conjecture about what others will do.\footnote{This holds even if the expressive component of the utility function is arbitrarily small (i.e., for any $b > 0$).} 

For those observing signals that do not give a dominant strategy, the optimal protest decision depends on the belief about how many other citizens are going to protest. This requires additional assumptions on the distributions of the incumbent popularity and the noise terms. To greatly simplify the analysis, we make the following distributional assumptions:
\begin{enumerate}
\singlespacing
	\item the incumbent popularity ($\omega$) is normally distributed with mean $\mu_0$ and precision $\tau_0$ (that is, variance $1/\tau_0$),
	\item the noise term in the election result ($\nu_e$) is normally distributed with mean 0 and precision $\tau_e$, 
	\item the noise terms in the private signals ($\nu_j$'s) is normally distributed with mean 0 and precision $\tau_{\theta}$, and
	\item all of the primitive random variables ($\omega$, $\nu_e$, and the $\nu_j$'s) are independent.\footnote{The assumption that the idiosyncratic component to the private signals ($\nu_j$'s)  is independent across individuals is potentially problematic. We might expect that these signals would be more strongly correlated among citizens that are ``close'' to each other, either geographically (especially when there is a strong regional component to politics, such as in Ukraine) or through social ties (e.g., some candidates may be more popular among adherents of particular religions); see \citet{dahleh2012global} for a global games model incorporating such social networks of information exchange.  Still, we do not see how a more complicated information structure would change our main conclusions, and hence proceed with the more tractable formulation.}
\end{enumerate}
\doublespacing

%% These caveats are important and the discussion is substantively interesting, but no one has given us a hard time about the joint normality, probably because it is so standard. It sort of slows things down a bit, so I wouldn't mind moving it to the discussion but I don't see a natural place for it now. Or we can just leave as is. 

%% JT: IF no one has bothered us about it, doesn't it make sense to just put all this in a footnote?  I've done this, but feel free to move back if you prefer.

The equilibrium condition is that when all other citizens use cutoff rule $\hat{\theta}(e)$, a citizen observing a private signal equal to the cutoff rule -- or, the \emph{marginal citizen} -- believes the expected proportion of protesters is just large enough to make them indifferent between protesting and not. Citizens observing a lower signal will get a higher expressive payoff for protest and believe that more citizens are going to protest, making protest optimal. Citizens observing a higher private signal will get a lower expressive payoff from protest and believe that fewer citizens will protest, making not protesting optimal. As derived in the appendix, the cutoff rule as a function of the election result $\hat{\theta(e)}$ must meet the following indifference condition:
\begin{align}
	\label{eqmfinal}
\underbrace{\frac {b \hat{\theta}(e) + c + k_1} {k_1}}_{\text{Indifference Protest Level}} = 	 \underbrace{\Phi\left(  \frac {(\tau_0 + \tau_e) \hat{\theta}(e) - \tau_0 \mu_0 - \tau_e (e - x)}
	{\tilde{\tau}^{-1/2} (\tau_0 + \tau_e  + \tau_{\theta})} \right)}_{\text{Expected Protest Level}}
\end{align}
where $\Phi(\cdot)$ is the cumulative density function (CDF) of a standard normal random variable and $\tilde{\tau}$ is a function of the primitive $\tau$ terms, defined in the appendix. The RHS of equation \ref{eqmfinal} represents the expected proportion of protesters for the marginal citizen with regime sentiment $\theta_j = \hat{\theta}(e)$. This is increasing in $\hat{\theta}(e)$ because higher thresholds mean citizens are more apt to protest. The LHS of equation \ref{eqmfinal} represents the expected level of protest required to make the marginal citizen indifferent between protesting and not. This is also increasing in $\hat{\theta}(e)$, as increasing the threshold means the marginal citizen likes the regime more (or dislikes the regime less). As elaborated below, for each election result there will be exactly one or three $\hat{\theta}(e)$ satisfying this condition and hence either a unique equilibrium or three potential equilibria. We first analyze the unique equilibrium case.

% The LHS of equation \ref{eqmfinal} is a increases linearly in the proposed equilibrium threshold ($\hat{\theta}(e)$)  with slope $\frac {b} {k_1}$ and the RHS as a function of the threshold is a normal CDF. As demonstrated below, if the precision of this CDF is sufficiently low there will be a unique intersection and hence a unique equilibrium in the protest stage for all election results. 

\section{The Unique Equilibrium Case and Semi-Democratic Turnover}
\label{sec.unique}
The size of protest as a function of the election result and other parameters has some intuitive comparative statics: 

\begin{proposition}
	\label{unique} When there is a unique equilibrium in the protest stage, the size of protest conditional on the election result $(\E[\rho|e])$ is:\\
 	i) decreasing in the cost of protesting against the incumbent ($c$ and $k_1$), and\\
	ii) decreasing in the election result ($e$)
\end{proposition}
\begin{proof}
	See the appendix.
\end{proof}

Increasing the cost of protest ($c$ or $k_1$)  has two effects. First, from the perspective of an individual citizen, increasing the cost terms for a fixed expected level of protest makes joining the protest less appealing. Second, knowing that other citizens experience the same higher cost means that they too are less likely to protest.  This in turn decreases any given individual's expectation about the overall level of protest, making one's own participation in the protest less favorable for reasons discussed previously.\footnote{This logic is related to the point made by \citet{T2007} that the manner in which electoral fraud leads to protest is by serving as a \emph{focal point} for many citizens who are dissatisfied with the government to have enough confidence that other citizens will simultaneously express their anger \emph{at this point in time} by taking to the streets.
  % In Tucker's formulation of the process, it is the existence of the publicized fraud that solves the collective action problem.  Here, we allow the election result itself to play that role through its effect on the estimated popularity of the incumbent.  In our formulation here, fraud is actually an impediment to measuring the true popularity of the incumbent, as we discuss below, but in both cases we are concerned with the process by which citizens become confident that their compatriots will join them on the streets.
}

The election result similarly has two effects. From the perspective of the incumbent (and analyst), lower election results indicate that the incumbent is less popular, and hence for a fixed citizen strategy more dislike the regime enough to protest. As the citizens observe the election result as well and expect more of their peers will protest, this leads to a higher willingness to protest. Note that both of these effects are entirely informational: the election result doesn't directly affect payoffs but provides information about what others are apt to do, potentially having a large affect on equilibrium behavior.

We now turn to the incumbent the incumbent decision to step down:

\begin{proposition}
	\label{uniquestepdown}
	If there is a unique equilibrium in the protest stage, the incumbent steps down if and only if the election result is sufficiently low, which happens with positive probability. The critical election result $e^*$ is given implicitly by $\E[\rho|e^*] = 1 - y^I$, where $y^I$ denotes the incumbent's payoff for stepping down.
\end{proposition}
\begin{proof}
	See the appendix
\end{proof}

%JT: Andrew, don't forget to get the citation about Chile and confirm no protests.  Has Adam written anything on this?

Thus the incumbent steps down if and only if the election result is sufficiently low that the ensuing protests make stepping down more appealing than standing firm. 
%% This is fairly different than last version. Some old paragraphs which I never particularly liked commented out, we can reinstate if you guys want
However, the unique equilibrium case does not capture alternation of power in a manner we would label as democratic, as the critical election result that determines whether the incumbent steps down may be completely unrelated to the legal threshold required for victory. In fact, other than for knife-edge cases, the critical election result \emph{won't} correspond to a formal rule. The critical threshold may happen to fall at a 50\% vote share in a given election, but even if so changes in any of the models parameters -- e.g., the cost of protest, or the prior incumbent popularity, which would surely be different in subsequent elections -- will shift this critical threshold away from the legally codified rule. As a result, we call this pattern ``semi-democratic'' alternation of power.

The fact that the incumbent steps down before the protest even begins is somewhat at odds with our motivating examples of semi-democratic turnover, although we do have real world examples of incumbents yielding power without being pushed to do so by post-election protests following elections in non-democratic countries; Poland in 1989 following the stunning and unexpected Solidarity victory is perhaps the clearest example \citep{A1999}.
% I think this example is just a bit too muddled - AL
% \footnote{Another example include's Pinochet's lost referendum on extending his rule in Chile in 1988 (CITATION NEEDED; confirm no protests). This example also highlights the fact that referenda and plebiscites can play the same role as elections in the model, but for conceptual clarity we restrict attention to elections with opposition.}
 Further, in section \ref{sec.beforeafter} we will examine an extension where the incumbent can step down before or after the protest, and in the analogous equilibrium they often will wait until after the initial protest to step down.

For an example an example where the incumbent did \emph{not} step down despite \emph{losing} by the official rules consider the 1991 Algerian parliamentary elections. The Islamist \emph{Islamic Salvation Front} won more than twice as many votes as any other party in the first round of these elections. By any democratic standards, the party would have taken control of the government following the second round of the election.  However, the military stepped in to prevent a change of government and cancelled the second round.  By the standards of our model, the first round of the election would have sent a signal that the incumbent was popular enough to remain in office, despite the fact that this level of popularity was nowhere near conventional democratic levels of majority support among the population.\footnote{We use the term ``majority'' loosely here as this was not a presidential election.  However, if we think of the indirect translation of ``majority" support in a parliamentary election as being able to win enough seats to form a government, then the forces supporting the Algerian military clearly were not going to win enough votes in this election to hold power in the parliament following the election.  
% At the same time, history reveals that the outcome of the election was not such that it caused the incumbent authorities to feel they had were so unpopular that they would need to give up power; as a countervailing example, witness the resignation of Eduard Shevardnaze in Georgia in 2003 as part of the Rose Revolution \citep{F2004}.  Thus our model would interpret this election result as having passed some threshold to convince the government it was popular enough to stay in power, but clearly this level of popularity was not related to what we would call a democratic or constitutional threshold.  
% As an aside, the Algerian case does suggest one potential future interesting avenue for research beyond our scope here, which is the role that international actors might be able to play in influencing the decision of the incumbent to stay in office or yield power.  One aspect of the Algerian case that might have been prevalent and would not be captured by our model would be the role of the West in providing support to a regime that would be seen as keeping Islamists out of power; from this vantage point, it may be fruitful to  contrast the outcome in Algeria with the last days of Hosni Mubarak's rule in Egypt two decades later.
} Thus the model highlights that winning the election by official rules was not the critical result needed for a transfer of power, though this may have happened if the Islamic forces has won even more convincingly.

\section{Multiple Equilibria and Electoral Rules}
\label{sec.multiple}

Recall that in solving the model, we first identify the citizens protest strategy for each election result and then determine whether the incumbent steps down or not as a function of the election result given this strategy. In the previous section, we analyzed the case where there is a unique protest strategy for each election result, which renders this process conceptually (if not algebraically) straightforward. However, there may be more than one strategy that the citizens can use in the protest stage for some election results: intuitively, one equilibrium where citizens are prone to protest when they think others are prone to protest and another equilibrium where citizens are prone to stay home if they think others are apt to stay home. 
% \tom{I think we should write 1-2 paragraphs restating the basic model, and stating exactly what ``multiple equilibria'' means here.}

To see that such multiplicity is possible, note that the equilibrium protest thresholds $\hat{\theta}(e)$ (from equation 2) are given by the intersection of an increasing line and a normal CDF. Some tedious algebra gives that the precision of this CDF is a function of the precisions of the primitive random variables:
\[		
	\tau_{RHS} =  \frac {\tau_{\theta} (\tau_0 + \tau_{e} + \tau_{\theta})} {2 \tau_{\theta} + \tau_{0} + \tau_{e}} \left(1 + \frac {\tau_{\theta}} {\tau_0 + \tau_e} \right)^{-2}
\]

\begin{figure}
	\caption{Illustration of Proposition \ref{multiple}. The dashed line corresponds to the LHS of equation \ref{eqmfinal}, and the curves to the RHS for various election results.}
	\begin{center}
		\includegraphics{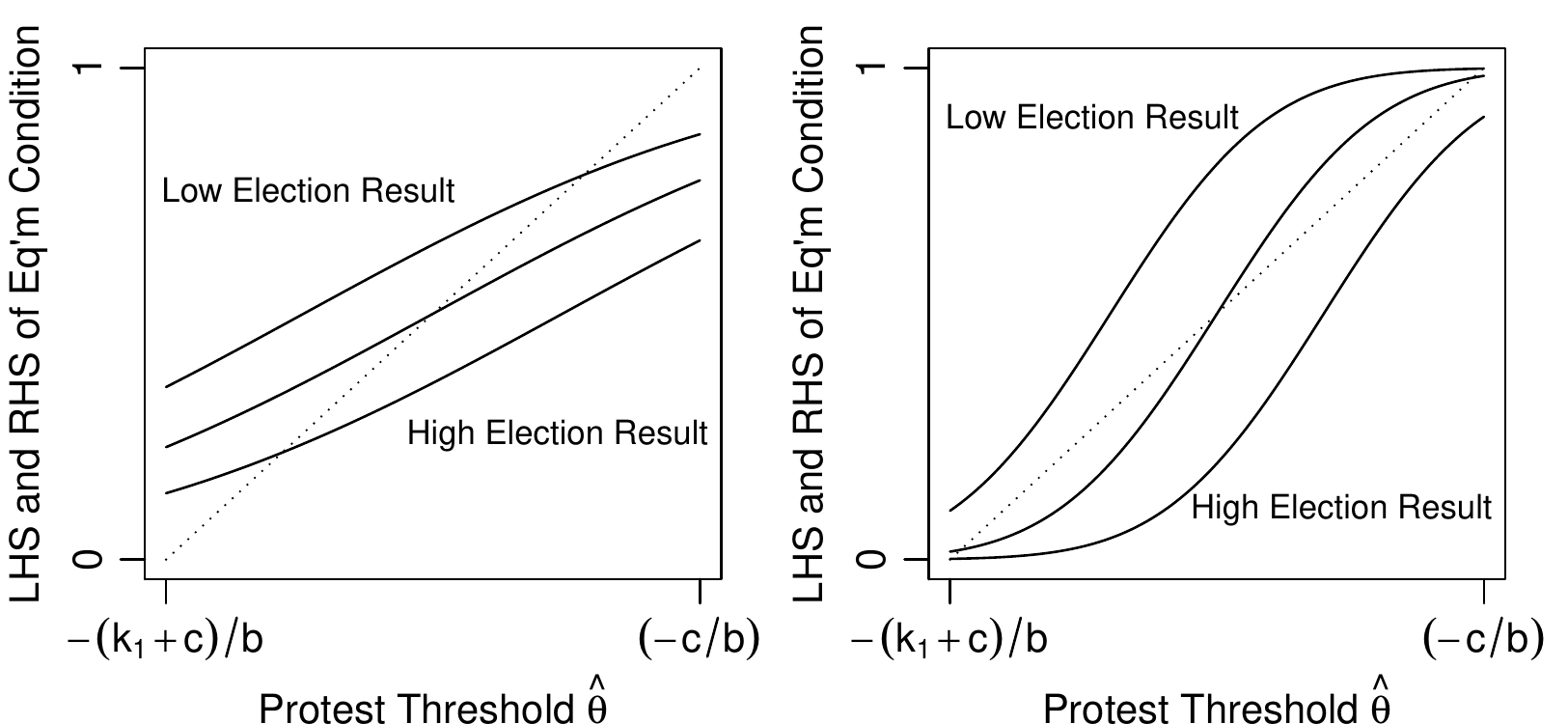}
		\label{intersections}
	\end{center}
\end{figure}

Figure \ref{intersections} plots this equilibrium condition for various election results for a case when $\tau_{RHS}$ is low (left panel) and high (right panel). The dashed line is the LHS of equation \ref{eqmfinal} and the solid curves the RHS for various election results, intersections correspond to equilibrium protest thresholds. In the left panel, the curves are always less steep than the line, meaning there is always a unique intersection and hence a unique equilibrium.  In the right panel, $\tau_{RHS}$ is higher, and for the middle curve -- corresponding to an intermediate election result -- there are three intersections and hence three potential equilibria in the protest stage.\footnote{Since the protest stage is only part of the game, we use the phrase ``multiple equilibria in the protest stage'' somewhat loosely. A more technically correct formulation is ``multiple threshold strategies that can be a part of an equilibrium in the protest stage for certain election results.''}

We analyze when there are multiple equilibria in more detail in the following section, but for now simply note that the patterns described above hold for any parameter choice. Either there is a unique equilibrium for all election results, or there is a unique equilibrium for extreme election results and three equilibria for intermediate election results.\footnote{There is also a knife-edge election result that results in two equilibria, but by the continuity of this random variable these election results happen with probability zero.} These potential equilibria can be ranked in terms of the expected level of protest: when the equilibrium protest threshold is higher, more citizens will dislike the regime enough to protest (recall citizens protest when their regime sentiment is below the threshold, so a higher threshold means more protest). As will be come clear when analyzing the incumbent decision to step down, we are mostly interested in equilibria with as much or as little protest as possible, and hence ignore the ``middle'' equilibrium and focus on the ``high'' and ``low'' protest equilibria.\footnote{Further, equilibria analogous to the middle protest threshold are generally termed unstable in the global games literature.} Formally:

% \tom{Does this mean, ``if there are multiple equilibria, then there are exactly two stable equilibria''? If so, that should be stated more clearly. This is really two claims: there are at least two stable equilibria, and there are at most two.}
% 
% \tom{Pulling a Scott, I would like 1-2 sentences explaining what stable equilibrium means. }

\begin{proposition}
	\label{multiple}
	There are at most three potential equilibrium strategies in the protest stage, and there are one or three potential equilibrium thresholds with probability 1. When there are multiple equilibria:\\
	i) The equilibrium with high levels of protests exists if and only if the election result is sufficiently low.\\
	ii) The equilibrium with low levels of protests exists if and only if the election result is sufficiently high\\
	iii) Within the highest or lowest protest equilibrium, the comparative statics in proposition \ref{unique} hold. For example, the level of protest is decreasing in the election result and the cost of protest.
\end{proposition}
\begin{proof}
See the appendix.
\end{proof}

Parts i and ii can be seen in the right panel of figure \ref{intersections}: for the low election result, there is only an equilibrium with a high protest threshold and hence a high level of protest, and for the high election result there is only an equilibrium threshold with a low level of protest. 

\subsection*{Stepping Down with Multiple Equilibria}

% As discussed above, the threat of protest can lead to the incumbent stepping down in the unique equilibrium case, but not in a manner that resembles compliance with electoral rules. However, suppose there are multiple equilibria for a wide range of election results, and the incumbent would not yield in the low protest equilibrium but would step down in the high protest equilibrium. Now, we can think of the equilibrium selection used by the citizens as something resembling democratic rules.
% 
% Thus an equilibrium where the citizens play the high protest equilibrium if and only if $e < e^*$ would resemble a democracy where the rules state that the incumbent wins the election if $e > e^*$. When $e < e^*$, the incumbent steps down. While this was the form of the unique equilibrium, a key distinction is that the critical election result $e^*$ can be placed at more than one election result, and, crucially, can correspond to an electoral rule as specified in a constitution. More importantly, the same equilibrium electoral rule can be enforceable even as other parameters (e.g., incumbent prior popularity, costs of protest, etc.) change. 

\begin{figure}
	\caption{Illustration of decision to step down in unique equilibrium case (left panels) and multiple equilibrium case (right panels). The bottom panels have a different cost of protest.}
	\begin{center}
		\includegraphics{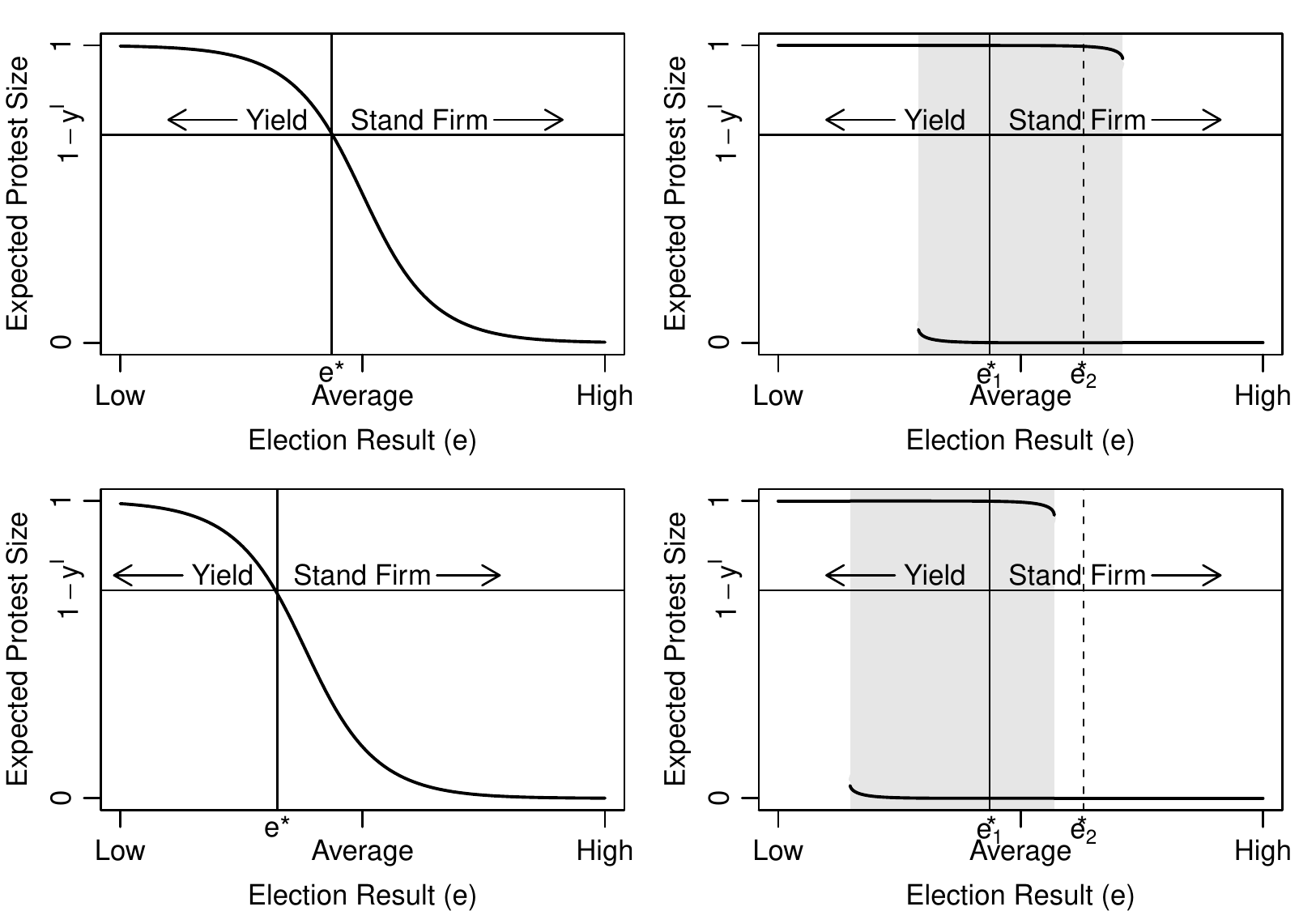}
		\label{uniqueandmult}
	\end{center}
\end{figure}

Now that we have a complete description of the potential equilibria in the protest stage, we can consider all possibilities in the incumbent decision to step down. Figure \ref{uniqueandmult} illustrates the distinction between the incumbent stepping down in the unique equilibrium and multiple equilibrium cases. In all panels, the curve traces the expected level of protest as a function of the election result. The incumbent steps down if and only if the curve is above $1 - y^I$, represented by the horizontal line. The left panels show the unique equilibrium case, and the right panels a similar multiple equilibrium case. For both, there is a higher cost of conflict in the bottom panel than the top panel.  

Starting with the unique equilibrium case, for both costs of conflict there is a unique $e^*$ such that the incumbent steps down if and only if $e < e^*$. However, this critical election threshold is lower when the cost of conflict is higher (bottom panel), indicating the incumbent stays in office while getting a lower vote share. As discussed above, we would not consider the unique equilibrium case democratic as the election result that determines who ``wins'' the election varies based on the exogenous parameters.

For the multiple equilibrium case, the shaded area indicates the election results for which there is a high protest equilibrium (upper line) and a low protest equilibrium (lower line). Since the high protest equilibrium is above $1 - y^I$ and the low protest equilibrium is below $1 - y^I$, the incumbent steps down if and only if the citizens play the high protest equilibrium for the observed election result.\footnote{If $y^I$ is close to 0 or 1 then it is possible that the incumbent does not step down in the face of a high protest equilibrium or does step down with a low protest equilibrium. This reduces the range of election results where the incumbent may or may not step down, but does not change the patterns described here.} So, for both costs there exists an equilibrium where the incumbent steps down if and only if $ e < e_1^*$. For the lower cost (top right panel), there is also an equilibrium where the citizens play the high protest equilibrium and hence the incumbent steps down if and only if $e < e_2^*$, where $e_2^* > e_1^*$. In fact, there is an equilibrium of this form for any $e^*$ in the shaded area. 

This is a key point: any electoral rule that falls in this range can be enforced by equilibrium behavior.  So, if there is a ``democratic" rule in this range (e.g., $e<.50$), then this rule can be enforced by equilibrium behavior according to our model, even thought our model posit that \emph{elections do no more than provide information about incumbent popularity.}  Moreover, this rule can remain enforceable even as the cost of conflict and other exogenous parameters change. The result is equilibrium behavior that \emph{mimics rule-based behavior} without turning elections into a contract in our model, and therefore a single model that explains alternations of power in both ``non-democratic" and ``democratic" 
% (i.e., systems that look like they are following rule-based contracts regarding executive turnover) 
systems.

However, note that $e^*_2$ is outside of the shaded range for the higher cost of conflict. So, if $e^*_2$ corresponds to the electoral law, then in the stylized example in Figure \ref{uniqueandmult}, the increase in the cost of protest going from the top to bottom panel would render this rule no longer enforceable.  This allows the model to capture the fact that if exogenous conditions shift too much -- e..g, the cost of protest increases too dramatically, something we might associate with a country becoming less democratic -- then the apparently ``rule-based" democratic behavior in equilibrium will eventually collapse.
% The range of potential equilibria also implies that a fixed electoral rule can be enforced across a range of exogenous parameters (e.g., incumbents of different levels of popularity; protest that is more or less costly). This is illustrated in the bottom panel of figure \ref{uniqueandmult}, which has the same parameters as the middle panel except a different prior belief about the incumbent popularity. This shifts the range of potential equilibrium electoral rules. However, both $e_1^*$ and $e_2^*$ are still feasible rules as they both remain in the shaded region. 

%JT: I think we can cut the following paragraph

%Another way of making this point is that many (if not most) contemporary regimes have a rule that says -- setting aside potentially important differences in detail -- that the candidate attaining the most votes wins the election, and in democracies this rule is consistently followed despite considerable political, economic, and social change. Thus, in this section we have identified a condition for this to be possible: the multiplicity of enforceable electoral rules must be large enough for the most-votes-wins rule to be an equilibrium for a wide range of these changes.

\subsection*{Informative Elections, Polarization, and Consolidated Democracy}

Having established the centrality of multiplicity for constructing equilibria where democratic rules are followed, we now examine in more detail \emph{when} there are multiple equilibria in the protest stage. Given the size of the parameter space and complexity of the equilibrium condition, there are many potential results, all of which require caveats. We focus on two results that we believe are particularly easy to interpret. More informative elections (high $\tau_e$) and a less polarized citizenry (low $b$) tend to lead to multiple equilibria for some election results.

A common thread in both results is that there tend to be multiple equilibria when citizens condition their behavior heavily on whether they expect other citizens to protest. When the election is uninformative, it is harder to form conjectures about what the election result means and hence whether other citizens dislike the regime enough to protest, so citizens will tend to rely more heavily on their personal regime sentiment. Similarly, when polarization is high, many citizens care more about the expressive benefit relative to considerations of whether others are going to take to the streets. So, uninformative elections and high polarization tend to be associated with a unique equilibrium. 

On the other hand, when the election is informative and citizens are not polarized, they condition their behavior more on their expectations of what others will do. This can result in multiple equilibria: one with a self-fulfilling expectation that many others will go out and protest and one with a self-fulfilling expectation that few citizens will protest:

\begin{proposition}
	\label{informative}
		i) There are multiple equilibria in the protest stage for some election results if and only if the polarization of citizens is sufficiently low. \\
	ii) Provided there is always a unique equilibrium for a completely uninformative election, there are multiple equilibria in the protest stage for some election results if and only if the election is sufficiently informative.
\end{proposition}
\begin{proof}
See the appendix.\footnote{The appendix also contains some technical caveats related to this result; for example, it is not necessarily the case that the probability of an election result with multiple equilibria is increasing in the precision of the election result.
  % In addition, recent papers in the global games literature identify other reasons why they may exhibit multiple equilibria in addition to high levels of public information (e.g., \citealt{AHP2007, BDM2011}). Still, the result above identifies conditions that we believe will generally lead to multiple equilibria in the protest stage.
}
\end{proof}
These results can be seen visually by referring back to figure \ref{intersections}. Increasing the precision of the election result increases $\tau_{RHS}$, which as shown in the contrast between the right and left panels makes it more possible to have multiple intersections and hence multiple equilibria. Similarly, lowering the polarization ($b$) makes the dashed line flatter, also raising the possibility of multiple intersections.

One way to interpret part ii is that following a free and fair election, voters can be relatively confident that the announced election result is a good proxy for the actual popularity of the incumbent. Further, given the public nature of the election result, citizens have a good conjecture about what \emph{other} citizens think about the incumbent popularity. The more norms of free and fair elections are violated, the less likely it is that all citizens will be able to come to the same conclusions regarding the underlying popularity of incumbents. It is this common belief about what the election result means that makes it easier for citizens to coordinate on either a high level or a low level of protest, enabling them to enforce electoral rules. 

Before turning to the extensions, we briefly address the question of why an equilibrium resembling compliance with a majoritarian electoral rule is particularly appealing equilibrium.
% \footnote{Given the centrality of multiplicity of equilibria for the argument made here, it would be inappropriate to apply equilibrium selection techniques that attempt to refine away all but a single equilibrium.} 
One argument is that following an electoral rule, particularly one at a natural threshold like a majority vote share is a natural focal point in sense coined by \citet[ch. 3]{S1960}. Citizens could threaten to protest against leaders who don't achieve the electoral threshold plus five percent, or on any publicly observed ``sunspot'', but these simply seem less natural than coordinating against law breakers. Further, given the fact that protest is costly and citizens want to protest against unpopular rules but not popular ones, such a rule may be optimal by putting popular leaders in office peacefully.\footnote{A full technical analysis of the optimal equilibria quickly becomes complex. Further issues like making fraud endogenous could complicate this argument, but we leave a thorough analysis of these questions to future work.}

\section{Unknown Fraud and Monitoring Reports}
\label{sec.monitor}
Until this point we have assumed that the amount of fraud is known.  However, this is rarely -- if ever -- the case. Indeed, in many of our motivating examples, attempts to resolve uncertainty about how much fraud was committed played a key role in the post-election protests. Further, we have some idea about how citizens come to hold ideas about how much fraud was committed, namely due to the publication of what we are terming \emph{monitoring reports}, by which we mean any public signals about how much fraud was committed. This could include actual reports issued by election monitors such as the Carter Center or Organization for Security and Cooperation in Europe, but also is intended to be broad enough to refer to news media reports of fraud, information (such as videos) posted to the internet, or reports of parallel vote tabulations. 

To take account of these considerations, we make two modifications to the model. First, we now interpret the ``noise'' in the election result to reflect uncertainty about how much fraud was committed. That is, $x$ now represents the average amount of fraud, and $x + \nu_e$ the ``true'' amount of fraud.\footnote{All the results here would hold as long as some of the uncertainty built into the election result comes from uncertainty about fraud.}
Second, in addition to the election signal, citizens now also observe a monitoring report $m$, which is a noisy signal of the level of fraud. The monitoring report is normally distributed with mean $x + \nu_e$ and precision $\tau_m$. That is, $m = x + \nu_e + \nu_m$ where $\nu_m$ is independent of the other noise terms and normally distributed with mean 0 and precision $\tau_m$. I
% n words, the main implication of these modifications is that a higher monitoring report will lead citizens to believe that there was more fraud than expected and hence the incumbent is less popular than expected. 

As in the main model, the joint normality assumptions lead to a convenient characterization of the citizens' \emph{interim} beliefs about the incumbent popularity, which summarize the public information citizens have about the incumbent popularity before observing their private signal. In the baseline model this is normal with a mean that is a weighted average of the prior mean ($\mu_0$) and election result less the expected level of fraud ($e - x$), and a precision that is the sum of the precision of the prior and election result. With the monitoring report, this belief is again a weighted average, with an additional correction term that reflects whether the monitoring report indicates more or less fraud than expected ($m - x$) and a higher precision ($\tau_0 + \tau_e + \tau_m$). Thus the analysis of the model is essentially the same with some additional terms added to the beliefs. The main implications of adding the monitoring report are:

% As elaborated, the results about how the election result affects the citizen and incumbent strategies can all be re-cast in terms of how the election result increases the interim belief about the incumbent popularity, which in turn affects the actors behavior. Similarly, monitoring reports indicating higher levels of fraud decrease the interim belief about the incumbent popularity, which gives the following:

\begin{proposition}
	\label{mcompstatics}
	If there is a unique equilibrium in the protest stage, the level of protest is increasing in the monitoring report (i.e., as the report claims more fraud), and the incumbent steps down if and only if the monitoring report claims sufficiently high levels of fraud.\footnote{Analogous results can be derived in the multiple equilibrium case.}
\end{proposition}

% Kinda long, but nice that our results are consistent here
Just as stronger election results for the incumbent lead to less protest because they indicate the incumbent is popular, reports of more fraud lead to more protest because they indicate the incumbent is, for a fixed election result, less popular.\footnote{That is, this is not driven by citizens being angry at cheating incumbents (as in \citet{T2007}).} This prediction is consistent with recent empirical work \citep{HM2012, R2012}. This extension is also consistent with a more nuanced result in \citet{R2012}, which finds that the effect of public reports of fraud on protest is larger when the margin of victory is small. In our model, the expected level of protest is an backwards (i.e., decreasing) s-shaped curve in the posterior belief about the incumbent popularity, which is a function of both the election result and monitoring report. When the election result is very high, and expected protest low, the posterior belief on the flat part of the s-curve and increasing beliefs about fraud has little effect on protest. However, for closer elections the posterior belief may be closer to the sharp part of the s-curve, meaning changes in monitoring reports can have a large impact on the resulting amount of protest.

Further, a similar result about the relationship between the precision of the election result and the possibility of multiple equilibria holds for the precision of the monitoring report:

\begin{proposition}
	\label{infomonitor}
	Provided there is a unique equilibrium without a monitoring report, there are multiple equilibria for some election results with a monitoring report provided the report is sufficiently informative.
\end{proposition}

The reason for this is straightforward given the previous result: one role that the monitoring report plays is to make the election result more informative as there is less uncertainty about the level of fraud. Thus institutions that detect electoral fraud can help consolidate democracy for a potentially unexpected reason: increasing the amount of public information and facilitating the coordination dynamics required to make electoral rules enforceable.\footnote{ \citet{SC2012} make a related argument that information generated by third parties like international monitors can help peaceful compliance with election results by alleviating information asymmetries between elites. Our results show that information can also 
facilitate peaceful transitions of power by alleviating information problems among the citizenry.}

\section{When to Step Down}
\label{sec.beforeafter}
The main model presented here assumes the incumbent has only one chance to step down, and that this chance is before the protest. As with any assumption, this is clearly unrealistic: leaders can concede defeat on the night of the election (as typically happens in democratic elections) or after days or weeks of protest (as in the Colored Revolutions).

To capture this difference, we extend the model to allow the incumbent to either step down before or after the protest. As in the main model, if the incumbent steps down right away, they get a payoff $y^I$; if they never step down the payoff is $1 - \rho$. If they step down after the protest, they get payoff $y_I - \gamma \rho$, where $\gamma \in (0, 1)$. That is, (1) when stepping down after the protest the incumbent is hurt by the level of protest, but not as much as they are if they do not step down at all, and (2) it is always better to step down before the protest than after the protest. Despite the second part, the incumbent will sometimes step down after the protest because they are unsure of how many citizens will take to the streets, and hence it can be valuable to wait things out and only step down if the protest is indeed large.
% (which is indeed what we often observe). 

The intuition behind how this modification affects the results is straightforward but the technical details involved are complex, so we relegate most of the formalization to the appendix. The first point is that citizen behavior at the protest stage is unaffected, so we only need to determine the incumbent yielding decisions. As above, the relative payoff from stepping down immediately diminishes the higher the vote for the incumbent candidate, so there will be a critical election result such that the incumbent steps down immediately if $e < e^*$.
% \footnote{This is assuming that in the multiple equilibrium case the citizens play an equilibrium of the form ``high protest if and only if the election result is sufficiently low,'' as was considered above.}  
If the incumbent does not step down, there will be a critical protest level $\rho^* \in (0,1)$ such that the the incumbent steps down following the protest if and only if $\rho > \rho^*$.

This pattern result holds regardless of whether or not there are multiple equilibria in the protest stage. However, there are crucial differences between the cases. In the unique equilibrium case, the value of standing firm is continuously decreasing in the election result, so there will be election results where a moderate level of protest is expected. Further, these are precisely the election results where the incumbent is uncertain about exactly how big the protests will be, giving an incentive to wait things out and see how large the protests become. This is well reflected in elite behavior in semi-democratic (or competitive authoritarian) regimes, where we have many examples of leaders attempting to hold onto office following a poor electoral showing that they know is likely to produce a high degree of protest, sometimes successfully (Iran, Russia) and sometimes unsuccessfully (Ukraine, Serbia). 
% According to our model, such behavior occurs because there is enough uncertainty in how much protest there will be to make it worth waiting things out to see exactly how many citizens actually do take to the streets. 

On the other hand, in the multiple equilibrium case there is a sharp discontinuity in the payoff to standing firm at the critical election result where the citizens switch from the high protest to low protest equilibrium. Further, as shown in figure \ref{uniqueandmult}, the expected level of protest tends to be very close to zero or one: meaning nearly none or all of the citizens will protest based on whether the electoral threshold is reached. So, there is little gain to waiting and seeing the actual protest size, which will means the incumbent will step down immediately after ``losing'' the election based on a codified rule, and faces minimal protest (and hence almost never steps down) upon winning.

So, there is an important difference between the multiple equilibrium case and unique equilibrium case in terms of \emph{when} the incumbent tends to step down. In the multiple equilibrium case -- which we have already argued resembles democratic turnover in office -- the incumbent tends to step down right after the election (upon losing), and large protests rarely happen. In the unique equilibrium case -- which again we label semi-democratic -- it is possible that the incumbent steps down right away, but will also very often not step down right away in order to see the realized level of protest.

\section{Conclusion}
\label{sec.conclusion}

The spread of elections globally is perhaps the most interesting development in comparative politics over the past half-century.  Nearly every country on earth now holds elections, and in the vast majority of cases there is at least the appearance of some competition; that is, most of these are not Soviet style 99\% turnout and 99\% incumbent vote share elections.  And yet, as social scientists (and theorists) we tend to still want to dichotomize the electoral experience into ``free" elections and (insert adjective here)\footnote{E.g., not free and fair, quasi-democratic, semi-democratic, competitive authoritarian, unequal, etc.} elections, with the former being imbued with an almost mythical contractual power that guarantees results will be respected and the latter treated as a separate animal if not dismissed as meaningless.  Somewhat surprisingly, as a field we lack a general theory of elections.  This paper is a step in that direction.

More specifically, the model presented above provides a general framework to capture how elections across regime types can facilitate alternation in power. In particular, it shows that a minimalist assumption about the role of elections -- that they are a public signal of the incumbent popularity -- can lead to equilibria that are qualitatively consistent with transfer of power due to (the threat of) post-election protest as well as peaceful and democratic transfer of power. This highlights two factors that tend to facilitate the possibility of democratic transfer of power: citizens that are not too polarized and elections that are informative.

More broadly, the results speak to the value of modeling elections primarily as information generating institutions. Here we have shown that this allows us to capture a phenomenon that at first may seem unrelated to information generation: peaceful and rule-based alternation of power. Further, the universality of this approach may build a bridge between the study of democratic and less-than-democratic elections.

\section*{Appendix A: Derivation of Equilibrium Condition, Proofs, and Technical Comments}

\subsection*{Conditional distributions of $\omega$} Here we derive the posterior belief about $\omega$ given $e$, $m$, and $\theta_j$, the other posteriors follow from a similar calculation. The joint distribution of $\omega$, $e$, $m$, and $\theta_j$ is a multivariate normal with mean vector $(\mu_0, \mu_0 + x, x, \mu_0)$ and covariance matrix:
\begin{align*}
	\Sigma = \bordermatrix{ & \omega & e & m & \theta_j \cr
		\omega & \tau_0^{-1} & \tau_0^{-1} & 0 & \tau_0^{-1} \cr
		e & \tau_0^{-1} & \tau_0^{-1} + \tau_e^{-1} & \tau_e^{-1} & \tau_0^{-1} \cr
		m & 0 & \tau_e^{-1} & \tau_e^{-1} + \tau_{m}^{-1} & 0 \cr
		\theta_j & \tau_0^{-1} & \tau_0^{-1} & 0 & \tau_0^{-1} + \tau_{\theta}^2}
	= \begin{pmatrix}
			\Sigma_{11} & \Sigma_{12}\\
			\Sigma_{21} & \Sigma_{22}
			\end{pmatrix}
\end{align*}
where $\Sigma_{11} = \tau_0^{-1}$ (which uniquely determines the remainder of the partition).  So, the desired posterior is normal (see Greene 2009; p 1014) with mean
\begin{align*}
	\ol{\ol{\mu}}(e,m, \theta_j) &= \mu_0 + \Sigma_{12} \Sigma_{22}^{-1}
			(e - (\mu_0 - x), m - x, \theta_j - \mu_0)\\
			&= \frac {\tau_0 \mu_0 + (\tau_e + \tau_m)(e - x) + \tau_{\theta} \theta_j - \tau_m (m - x)} {\tau_0 + \tau_e + \tau_m + \tau_{\theta}}
\end{align*}
and precision
\[
	(\Sigma_{11} - \Sigma_{12}\Sigma_{22}^{-1}\Sigma_{21})^{-1/2} = \tau_0 + \tau_e + \tau_m + \tau_{\theta}.
\]

\subsection*{Derivation of the Equilibrium Condition}

The indifference condition for a citizen with the cutoff level of regime sentiment $\theta_j = \hat{\theta}(e)$ is:
\[
	\hat{\theta}(e) = \frac {-k_1 - c + k_1 \E[\rho|\hat{\theta}(e); \theta_j = \hat{\theta}(e)]} {b}
\]
What remains to be derived is $\E[\rho|\hat{\theta}(e); \theta_j = \hat{\theta}(e)]$.   The probability of receiving a signal that will lead to protest is $Pr(\theta_j < \hat{\theta}(e))$, so the expected proportion of protesters for as a function of the true popularity $\omega$ and proposed equilibrium threshold is:
\begin{align*}
	\E[\rho|\omega, \hat{\theta}(e)] =  Pr(\theta_j < \hat{\theta}(e)) =
		\Phi\left(\tau_{\theta}^{1/2} (\hat{\theta}(e) - \omega) \right)
\end{align*}
Next we need to derive the belief about $\omega$. Two of these beliefs used below are the \emph{mean interim belief} only conditioned on the public information (here, the election result), which is normal with mean
\begin{align*}
	\ol{\mu}(e) \equiv \frac {\tau_{0} \mu_0 + \tau_e (e - x) }
			{\tau_0 + \tau_e}
\end{align*}
and precision $\tau_0 + \tau_e$. Note that $e \mapsto \ol{\mu}(e)$ is an linear function of $e$, and is invertible except in the degenerate case when $\tau_e = 0$. The posterior belief about $\omega$ given $e$ \emph{and} the private signal $\theta_j$ is normally distributed with mean
\begin{align*}
	\ol{\ol{\mu}}(e,\theta_j) &\equiv \frac {(\tau_0 + \tau_e) \ol{\mu}(e) + \tau_{\theta} \theta_j} {\tau_0 + \tau_e + \tau_{\theta}}
	=  \frac {\tau_0 \mu_0 + \tau_e (e - x)
		 + \tau_{\theta} \theta_j }
			{\tau_0 + \tau_e + \tau_{\theta}}
\end{align*}
and precision $\tau_0 + \tau_e + \tau_{\theta}$. 

Let $\phi(\omega; \mu, \tau)$ be the probability density function of $\omega$ when the belief is normal with mean $\mu$ and precision $\tau$. Then $\E[\rho|\theta_j; \hat{\theta}(e), e]$ becomes:
\begin{align*}
	\E[\rho|\theta_j; \hat{\theta}(e), e] &=	\int_{-\infty}^\infty \Phi\left( \tau_{\theta}^{1/2} (\hat{\theta}(e) - \omega) \right)
	  	\phi(\omega; \ol{\ol{\mu}}(e,\theta_j), \tau_0 + \tau_e + \tau_{\theta}) d \omega\\
		&= \Phi\left( \tilde{\tau}^{1/2} (\hat{\theta}(e) - \ol{\ol{\mu}}(e,\theta_j)) \right)
\end{align*}
Where $\tilde{\tau} = \frac {\tau_{\theta} (\tau_0 + \tau_{e} + \tau_{\theta})} {2 \tau_{\theta} + \tau_{0} + \tau_{e}}$.\footnote{The equality follows from the fact that a normally distributed variable where the mean is unknown and normally distributed is itself normally distributed.} Evaluating this belief for the marginal citizen ($\theta_j = \hat{\theta}(e)$) and plugging into the indifference condition gives the equilibrium condition in the main text.

% \tom{Removed Proof of Lemma 1 since it was one line, and moved content to the actual proof of Lemma 1}

%\subsection*{Proof of Lemma 1}
%This follows from the fact that  $\ol{\mu}(e)$ is a linear (and invertible) transformation of $e$, there is a one to one mapping between the election result and interim mean belief.

\subsection*{Proof of Proposition 1}

The expected number of protestors (given the election result $e$) is
\begin{align*}
		\E[\rho|e] &=	\int \Phi\left(\tau_{\theta}^{1/2} (\hat{\theta}(e) - \omega) \right) \phi(\omega; \ol{\mu}(e), \tau_0 + \tau_e) d\omega\\
		&= \Phi\left(\tau_{\rho}^{1/2}(\hat{\theta}(e) -
		\ol{\mu}(e) \right)
\end{align*}
where $\tau_{\rho} = \frac {\tau_{\theta} (\tau_0 + \tau_e)} {\tau_{\theta} + \tau_0 + \tau_e}$. This is now in a convenient form for considering comparative statics. The derivative of the number of protesters with respect to the election result is:
\begin{align*}
	\frac {\partial \E[\rho|\cdot]} {\partial e} &=
		\phi(\cdot) \tau_{\rho}^{-1/2} \left(- \frac {\tau_0 + \tau_e} {\tau_0 + \tau_e + \tau_{\theta}}+ \frac {\partial \hat{\theta}} {\partial e} \right)\\
		&= 	\phi(\cdot) \tau_{\rho}^{-1/2} \left(\frac {\tau_0 + \tau_e} {- \tau_0 + \tau_e + \tau_{\theta}} + \left(-\frac {1/(k_1)} {(b/(k_1)) - \phi(\cdot)\left( \frac {\tilde{\tau} \tau_0} {\tau_0 + \tau_e} \right)} \right) \right) < 0
\end{align*}
(We suppress the argument of the normal pdf since all that matters for the sign is that it is positive.) 
	Part iii follows from implicitly differentiating the equilibrium condition with respect to $c$ and $k_1$, and the fact that the LHS is increasing faster than the RHS at any intersection corresponding to a stable equilibrium. 

\subsection*{Proof of proposition 2}
As derived in the main text, the expected level of protest given $e$ is
\begin{align*}
	\E[\rho|e] &= \Phi\left(\tau_{\rho}^{1/2} \left(\hat{\theta}(e) -
		\frac { \tau_0 \mu_0 + \tau_e (e - x)} {\tau_0 + \tau_e}\right) \right)
\end{align*}
From inspection, this expectation is increasing in $e$. Further, it follows from proposition 1 that $\lim_{e \to \infty} \E[\rho|e] = 0$ and $\lim_{e \to -\infty} \E[\rho|e] =  1$. Since $y^I \in (0,1)$, $Pr(\E[\rho|e] < 1 - y^I) \in (0,1)$. As we can write the equilibrium condition as a function of $\ol{\mu}(e)$, we can write the expected level of protest as:
\begin{align*}
	\E[\rho|\ol{\mu}(e)] &= \Phi\left(\tau_{\rho}^{1/2}(\hat{\theta}_\mu(\ol{\mu}(e)) - \ol{\mu}(e)\right)
\end{align*}
Define $\mu_y$ to be the interim mean belief such that $\E[\rho|\mu_y] = 1 - y^I$, hence the incumbent steps down if and only if $\ol{\mu}(e) < \mu_y$. Setting $e^* = \ol{\mu}^{-1}(\mu_y)$ completes the proof. $\qed$

\subsection*{Proof of propositions 3-4}

To prove these results, write the equilibrium condition as:

\begin{equation}
	\label{eqmab} \alpha + \beta \hat{\theta} = \Phi\left( \frac {\hat{\theta} - \mu_{RHS}} {\sigma_{RHS}}\right)
\end{equation}

Where $\alpha = \frac {c + k_1} {k_1}$, $\beta = \frac {b} {k_1}$, $\mu_{RHS} = \ol{\mu}(e)$, and $\sigma_{RHS} = \frac {\sqrt{\tilde{\tau}} (\tau_0 + \tau_e)} {\tau_0 + \tau_e + \tau_{\theta}}$. We first present intermediate results with this less cumbersome form, and then the results in the main text quickly follow.

\begin{lemma} \label{lem_intersections}
	i. If $\beta > 1/(\sigma_{RHS})$, there is a unique $\hat{\theta}$ satisfying equation \ref{eqmab}.\\
	If $\beta < 1/(\sigma_{RHS} \sqrt{2 \pi})$, then there exists $\mu_l, \mu_h \in \R$, $\mu_l < \mu_h$, such that:\\
ii. If $\mu_{RHS} < \mu_l$ or $\mu_{RHS} > \mu_h$, there is a unique solution to equation \ref{eqmab}.\\
iii. If $\mu_{RHS} \in (\mu_l, \mu_h)$, there are three solutions to equation \ref{eqmab}.\\
iv. If $\mu_{RHS} = \mu_l$ or $\mu = \mu_h$ there are two solutions to equation \ref{eqmab}
\end{lemma}
\begin{proof}
Figure \ref{intersections} in the main text shows the result graphically, a formal proof is available upon request. If the RHS has a low standard deviation as in the left panel, the maximum slope will be less than the slope of the LHS, guaranteeing a unique intersection. If the RHS has a high standard deviation (as in the right panel), there will be multiple intersections if the steep section of the CDF intersects the LHS from below, which occurs for intermediate $\mu_{RHS}$. Since the first derivative of the RHS of equation \ref{eqmab} is a constant and the derivative of the LHS is a normal PDF, the difference between the RHS and the LHS can only change signs twice, which implies there can be no more than three intersections. $\qed$
\end{proof}

Using Lemma \ref{lem_intersections}, we can quickly prove the results in the main text:

\textbf{Proof of proposition 3.} 

The number of equilibria follows from the number of intersections derived above, and the fact that $\mu_{RHS}$ is exactly equal to $\mu_h$ or $\mu_l$ with probability zero. This also implies that the ``low protest equilibrium'' lies lies on the convex part of the RHS (i.e., $\hat{\theta} < \mu_{RHS}$) and the ``high protest equilibrium'' lies on the concave part of the RHS (i.e., $\hat{\theta} > \mu_{RHS}$). Part ii and iii then follow from the definition of $\mu_{RHS}$ and comparing the signs of the derivatives of the LHS and RHS at the intersections as in the proof of the lemma. Part iv follows from implicitly differentiating the equilibrium condition and the fact that the LHS must be increasing faster than the RHS for these equilibria.  $\qed$

\textbf{Proof of proposition 4.} Part i of proposition 4 follows from the fact that $\beta <  1/(\sigma_{RHS})$ is equivalent to $b < \frac {k_1} {\sigma_{RHS}}$, so there are multiple equilibria for some $\mu_{RHS}$ if and only this condition is met. Since any value of $\mu_{RHS}$ can be met for some $e$, there will be multiple equilibria for some election result for $b$ below the critical value. Part ii follows a similar logic: if $\sigma_{RHS}$ is sufficiently low as $\tau_e \to 0$, there will be multiple equilibria for any $\tau_e$. If not, since $\sigma_{RHS} \to 0$ as $\tau_e \to \infty$, there will exist a critical $\tau_e^*$ such that the multiple equilibrium condition is met if and only if $\tau_e > \tau_e^*$.

\subsection*{Discussion of Proposition \ref{informative}}

% \tom{Why is this section in the appendix and not in the text? I recommend moving it to the main text, and labeling it ``Technical Discussion'' or something like that, along with a single sentence stating that the reader may skip it.}

We would like to draw attention to several caveats to proposition \ref{informative}. First, the more nuanced language in the result about the informativeness of elections is because their may be multiple equilibria even with a completely uninformative election, hence a precise election result is not required for the focal point effect. If this is the case, then there can be equilibria exhibiting the focal point effect even for arbitrarily noisy (i.e., uninformative) election results

Second, it is not necessarily the case that the probability of there being multiple equilibria  is necessarily decreasing in $b$ or increasing in the precision of the monitoring report. So, it is possible that increasing the precision of the monitoring report makes a high protest equilibrium \emph{less} likely for some range of $\tau_e$.

Finally, it is well known that, both in the global games context and in general, adding realistic components such as giving the incumbent private information or repeating the protest stage or entire game will often lead to multiple equilibria even if the baseline model has a unique equilibrium.\footnote{\citet{AHP2006} and \citet{AHP2007} consider global games with signaling and in a repeated setting, respectively. In both cases, these modifications can lead to multiple equilibria.} Our only defense against this concern is that all (formal) theories must make simplifications to clarify potential relationships between ``independent'' variables of interest (here, the informativeness of elections and the polarization of the polity) and outcomes (here, the range of election results with multiple equilibria and hence the possible of rule-based alternation of power). The potential that other channels through which the independent variables affect the outcomes may confound, negate, or even reverse the relationships found in a model is inevitable. We admit it is particularly hard to assess the magnitude of this problem when dealing with abstract notions such as the pervasiveness of multiple equilibria, but see no concrete reason to worry that including signaling or repeated game dynamics would change our central results.

\subsection*{Proof of propositions 5-6}

The derivations in the appendix show how the equilibrium condition and other results can be derived with respect to the interim mean popularity $\ol{\mu}(e)$. With the addition of the monitoring report the interim mean belief is:
\[
	\ol{\mu}(e,m) = \frac {\tau_0 \mu_0 + (\tau_e + \tau_m) (e - x) + \tau_m (m - x)} {\tau_0 + \tau_e + \tau_m}
\]

Both are linearly increasing in $e$ and the belief with the monitoring report is linearly decreasing in $m$.
\begin{lemma}
	\label{fundamental}
	i. We can write the equilibrium threshold rule as the mean interim popularity:  $\hat{\theta}(\ol{\mu}(e))$ or  $\hat{\theta}(\ol{\mu}(e, m))$, and \\
	ii.  In both the baseline model and extension, all of the results referring to the election result can be rewritten with respect to the interim mean popularity
\end{lemma}
\begin{proof}
	Part i follows from doing an analogous derivation of the equilibrium condition when also accounting for the monitoring report. In particular, the citizen posterior belief about the incumbent popularity with the monitoring report and their private signal has mean: 
	\begin{align*}
		\ol{\ol{\mu}}(e,m, \theta_j) &\equiv \frac {(\tau_0 + \tau_e + \tau_m) \ol{\mu}(e, m) + \tau_{\theta} \theta_j} {\tau_0 + \tau_e + \tau_m + \tau_{\theta}} 
		=  \frac {\tau_0 \mu_0 + (\tau_e + \tau_m) (e - x) - \tau_m (m - x)
			 + \tau_{\theta} \theta_j }
				{\tau_0 + \tau_e + \tau_m + \tau_{\theta}}
	\end{align*}
and precision $\tau_0 + \tau_e + \tau_m + \tau_{\theta}$. This gives an analogous equilibrium condition for $\hat{\theta}(e,m)$, which is decreasing in $m$ in the unique equilibrium case (and the high and low protest equilibria with multiple equilibria).
	For part ii, it is clear from definition that $\ol{\mu}(e)$ is an invertible function of the election result $e$. So, as long as there is a unique equilibrium for all $e$ there will be a unique equilibrium for all $\ol{\mu}(e)$ and a one-to-one correspondence between these descriptions of the equilibrium threshold.\footnote{The uniqueness caveat is only necessary because when there are multiple equilibria the mapping from the election result to equilibrium thresholds is not a function.} So all comparative statics with respect to $e$ can be equivalently derived with respect to $\ol{\mu}(e)$ in the baseline model and $\ol{\mu}(e,m)$ with the monitoring report. $\qed$ 
%	This follows from writing the equilibrium condition with $\ol{\mu}(e)$ replacing the relevant terms, see the appendix for details.
\end{proof}

The proposition follows from the fact lemma \ref{fundamental} and the fact that $\ol{\mu}(e,m)$ is decreasing in $m$.

% \subsection*{Proof of proposition 9}
% \tom{This argument is WAAAAAY too brief. I have no idea why the single short line after Assumption 1 implies the stronger result of Proposition 9. Moreover, this assumption is not ``mild''-- it's technical and extremely strong, and needs justification. Maybe Josh can help with this one.}
% 
% The assumption placed on the join distribution of the private fraud signal is:
% \begin{assumption}
% 	For any $\theta^\prime  \in (\frac {-k_1 - c} {b}, \frac {- c} {b})$ and $p \in (0,1)$, there exists $\ul{\psi}, \ol{\psi} \in \R$ such that $Pr(\theta_j < \frac {-k_1 - c} {b}|\theta^\prime, \ul{\psi}) > p$ and  $Pr(\theta_j < \frac {-k_1 - c} {b}|\theta^\prime, \ol{\psi})  > p$.
% \end{assumption}
% So for $\theta_j \in (\frac {-k_1 - c} {b}, \frac { - c} {b})$, the citizens may receive a private signal extreme enough to ensure either protesting or not protesting is optimal $\qed$

\section*{Stepping Down Before of After Protest}

As in the main model, if the incumbent steps down right away they get a payoff $y^I$ and if they never step down the payoff is $1 - \rho$. If they step down after the protest, they get payoff $y_I - \gamma \rho$, where $\gamma \in (0, y^I)$. 

The citizen behavior at the protest stage is unaffected by this modification. If $I$ does not step down immediately, the choice is between yielding and getting payoff $y_I - \gamma \rho$ and not stepping down and getting payoff $1 - \rho$. Rearranging, $I$ steps down if and only if:
\begin{align}
	% y_I - \gamma \rho > 1 - \rho
	\rho > \frac { 1- y^I} {1 - \gamma} \equiv \rho^*
\end{align}

Now consider the decision to step down or not \emph{before} the protest stage. The expected payoff for not stepping down right away is:
\begin{align*}
	u^I_{SF}(e)  
		% = \int_{\rho = 0}^{\rho^*} (1 - \rho) f_{\rho|e} d \rho
				% + \int_{\rho=\rho^*}^1 (y_I - \gamma \rho f_{\rho|e}) d \rho\\
				= Pr(\rho < \rho^*|e) (1 - \E[\rho|\rho<\rho^*, e])
				 + Pr(\rho > \rho^*|e) (y^I - \gamma \E[\rho|\rho>\rho^*, e])
\end{align*}

For arbitrarily low election results, regardless of the citizen strategy the probability that there will be a high level of protest becomes arbitrarily large, hence $ Pr(\rho < \rho^*|e)$ approaches 1 and hence $\lim_{e \to -\infty} u^I_{SF}(e) = y^I - \gamma \E[\rho|\rho<\rho^*, e] < y^I$. So, for sufficiently low election results, the incumbent steps down right away. Similarly, for arbitrarily high election results the incumbent knows there will be virtually no protest and hence $\lim_{e \to -\infty} u^I_{SF}(e) = 1 > y^I$, hence for arbitrarily high election results the incumbent does not step down right away.  

Further, as long as equilibrium threshold rule is monotone in the multiple equilibrium case (i.e., low protest equilibrium if and only if the election result is sufficiently high), the payoff to standing firm initially is decreasing in the election result. So:

\begin{proposition}
	There exists an $e^*_{SF} \in \R$ such that the incumbent steps down immediately after the election if and only if $e < e^*_{SF}$
\end{proposition}
\begin{proof}
	This follows from the limiting behavior and the monotonicity of $u^I_{SF}(e)$ in $e$.
\end{proof}

The following result formalizes the idea that there tends to be more protest on the equilibrium path in the unique equilibrium case 
% (see the full appendix for more detail \tom{Again, there should be a pointer to exactly where in Appendix B the more details may be found.}):

\begin{proposition}
	When the polarization of citizens is arbitrarily small of the election arbitrarily informative,\\
	i) the incumbent always steps down immediately for election results with a high protest equilibrium, and \\
	ii) the probability of stepping down as a result of protest (after winning an election) approaches zero.
\end{proposition}
\begin{proof}
	Part i follows from the fact that as $b \to 0$ and as $\tau_e \to \infty$, $\E[\rho|e] \to 1$ and $Pr(\rho>\rho^*) \to 1$ for any $e$ where the citizens play a high protest equilibrium. Part ii follows from the fact that as $b \to 0$ and as $\tau_e \to \infty$, $\E[\rho|e] \to 0$ and $Pr(\rho>\rho^*) \to 0$ for any $e$ where the citizens play a low protest equilibrium. $\qed$
\end{proof}

\bibliographystyle{apsr}
\bibliography{ffbib}
\end{document}